\documentclass[preprint,aps,nofootinbib,superscriptaddress]{revtex4-1}
\usepackage{amssymb,mathrsfs,physics,amsthm,amssymb,amsfonts,amsmath}

\usepackage{hyperref}
\hypersetup{
    plainpages=false,       % needed if Roman numbers in frontpages
    unicode=false,          % non-Latin characters in Acrobat’s bookmarks
    pdftoolbar=true,        % show Acrobat’s toolbar?
    pdfmenubar=true,        % show Acrobat’s menu?
    pdffitwindow=false,     % window fit to page when opened
    pdfstartview={FitH},    % fits the width of the page to the window
    pdfnewwindow=true,      % links in new window
    colorlinks=true,        % false: boxed links; true: colored links
    linkcolor=Blue,         % color of internal links
    citecolor=Red,        % color of links to bibliography
    filecolor=magenta,      % color of file links
    urlcolor=blue           % color of external links
}
\usepackage[usenames, dvipsnames]{xcolor}

\newcommand\id{\leavevmode\hbox{\small1\kern-3.3pt\normalsize1}}

\usepackage{enumitem} 

\usepackage[normalem]{ulem}

\theoremstyle{definition}
\newtheorem{theorem}{Theorem}%[section]
\newtheorem{corollary}{Corollary}[theorem]

\newtheorem{proposition}[theorem]{Proposition}

\usepackage{bbm}

\usepackage{tikz}
\usetikzlibrary{arrows}

\usepackage{CJKutf8}

\begin{document}

\begin{abstract}
The quantum process framework is devised as a generalization of quantum theory to incorporate indefinite causal structure. In this short paper we point out a restriction in forming products of processes: there exist both classical and quantum processes whose tensor products are not processes. Our main result is a necessary and sufficient condition for when tensor products of multipartite processes are (in)valid. We briefly discuss the implications of this restriction on process communication theory.
\end{abstract}
	
\title{Process Products Need Qualification}
\begin{CJK*}{UTF8}{gbsn}
\author{Ding Jia (贾丁)}
\email{ding.jia@uwaterloo.ca}
\affiliation{Department of Applied Mathematics, University of Waterloo, Waterloo, ON, N2L 3G1, Canada}
\affiliation{Perimeter Institute for Theoretical Physics, Waterloo, ON, N2L 2Y5, Canada}
\author{Nitica Sakharwade}
\email{nsakharwade@perimeterinstitute.ca}
\affiliation{Perimeter Institute for Theoretical Physics, Waterloo, ON, N2L 2Y5, Canada}
\affiliation{Department of Physics and Astronomy, University of Waterloo, Waterloo, ON, N2L 3G1, Canada}
\date{\today}

\maketitle
\end{CJK*}

\section{Introduction}
The quantum process framework \cite{oreshkov2012quantum, araujo2015witnessing, oreshkov2016causal} is devised as a generalization of ordinary quantum theory to incorporate indefinite causal structure \cite{hardy_probability_2005, hardy2007towards, hardy2009quantum, chiribella2013quantum}. In this short paper we point out a restriction on the composability of processes by deriving a necessary and sufficient condition for when tensor products of multipartite processes are (in)valid.

\section{Processes}
In this section we very briefly recall the relevant part of the process framework \cite{oreshkov2012quantum, araujo2015witnessing, oreshkov2016causal} and introduce some nomenclatures. It is postulated that local physics is described by ordinary quantum theory with definite causal structure, while the global (possibly indefinite) causal structure is encoded in processes. A process is a linear map from CP maps describing local physics to real numbers describing probability of observation outcomes. Both channels and states are special cases of processes.

We use $A,B,C,\cdots$ to denote the parties where local physics takes place. A party $A$ is associated with an input system $a_1$ with Hilbert space $\mathcal{H}^{a_1}$ and an output system $a_2$ with Hilbert space $\mathcal{H}^{a_2}$. Through the Choi isomorphism \cite{choi1975completely} processes can be represented as linear operators. A process $W$ associated with parties $A,B,C,\cdots$ is represented as a linear operator $W^{a_1a_2b_1b_2c_1c_2\cdots}\in L(\mathcal{H})$, where $\mathcal{H}:=\mathcal{H}^{a_1}\otimes\mathcal{H}^{a_2}\otimes \mathcal{H}^{b_1}\otimes\mathcal{H}^{b_2} \otimes \mathcal{H}^{c_1}\otimes\mathcal{H}^{c_2}\otimes\cdots$. Sometimes we write the process as $W^{abc\cdots}$ for simplicity.

It is assumed that processes yield normalized and positive outcome probabilities, and can also act on subsystems of local parties. These imply that
\begin{align}
W\ge& 0\label{eq:Wfc1} 
\\
\Tr W=&\dim(\mathcal{H}^{a_2}\otimes\mathcal{H}^{b_2}\otimes\mathcal{H}^{c_2}\otimes\cdots)=:d_O,\label{eq:Wfc2}
\end{align}
and in addition:
\begin{proposition}[Oreshkov and Giarmatzi \cite{oreshkov2016causal}, reformulated using the language of this paper]\label{prop:og}
A multiple partite process obeys the normalization of probability condition if and only if in addition to the identity term it contains at most terms which is type $a_1$ on some party $A$.
\end{proposition}

%both local CP maps and processes can be represented as linear operators. A local CP map $\mathcal{M}$ from $a_1$ to $a_2$ describing the operation within party $A$ is represented as a linear operator $M_{a_1}^{a_2}\in L(\mathcal{H}^{a_2}\otimes\mathcal{H}^{a_1})$, while a process $W$ associated with parties $A,B,C,\cdots$ is represented as a linear operator $W_{a_2b_2c_2\cdots}^{a_1b_1c_1\cdots}\in L(\mathcal{H})$, where $\mathcal{H}:=\mathcal{H}^{a_1}\otimes\mathcal{H}^{a_2}\otimes \mathcal{H}^{b_1}\otimes\mathcal{H}^{b_2} \otimes \mathcal{H}^{c_1}\otimes\mathcal{H}^{c_2}\otimes\cdots$. Here we adopted the tensor-type notation of \cite{hardy2011reformulating, hardy2012operator} to denote inputs with subscript and outputs with superscript.

%, and composition of inputs with outputs with repetitive index. For example, a set of local operations $M_{a_1}^{a_2}$ 

To understand this proposition we need to introduce the notion of types. A process $W^{ab\cdots}$ can be expanded in the Hilbert-Schmidt basis $\{\sigma_i^x\}_{i=0}^{d_x^2-1}$ of the $x$ subsystem operators $L(\mathcal{H}^x)$ as
\begin{align}
W^{ab\cdots}=\sum_{i,j,k,l,\cdots} w_{ijkl\cdots} \sigma_i^{a_1}\otimes \sigma_j^{a_2}\otimes  \sigma_k^{b_1}\otimes  \sigma_l^{b_2}\otimes \cdots, ~ w_{ijkl\cdots}\in \mathbb{R}.
\end{align}
%For example, for qubit systems the Pauli basis contains four elements $\sigma_0=\id$, and the Pauli operators $\sigma_i, i=1,2,3$. For general systems, we always take $\sigma_0=\id$. 
We set the convention to take $\sigma_0^{x}$ to be $\id$ for any $x$. We refer to terms of the form $\sigma^{x}_i\otimes \id^{\text{rest}}$ for $i\ge 1$ as a type $x$ term, $\sigma^{x}_i\otimes \sigma^{y}_j \otimes \id^{\text{rest}}$ for $i,j\ge 1$ as a type $xy$ term etc. The identity term is referred to as a trivial term to be of trivial type.

Restricting attention to some party $A$, we say that a term is ``in type'' if it is $a_1$ type, and is ``out type'' if it is $a_2$ type. We say that it ``includes the in type'' if it is $a_1$ or $a_1a_2$ type, and ``includes the out type'' if it is $a_2$ or $a_1a_2$ type. On two parties $A$ and $B$, terms of type $a_2b_1$ and $a_1a_2b_1$ are called $A$ to $B$ signalling terms (since $A$'s output is correlated with $B$'s input), and similarly terms of type $a_1b_2$ and $a_1b_1b_2$ are called $B$ to $A$ signalling terms.

\section{Conditions for forming valid and invalid products}
In this section we introduce process products and prove our main results. A party $\{A',a_1',a_2'\}$ and a party $\{A'',a_1'',a_2''\}$ can be combined into a new party $\{A,a_1,a_2\}=\{A'A'',a_1'a_1'',a_2'a_2''\}$ if all channels from $a_1$ to $a_2$ can be applied by $A$. Here $A'A''$ is a shorthand notation for combining parties $A'$ and $A''$, and $xy$ is a system whose Hilbert space is $\mathcal{H}_x\otimes \mathcal{H}_y$.

For resources in quantum theory with definite causal structure, there is no restriction in taking tensor products. A channel $M$ is a two-party resource that mediates information between some party $A'$ and some party $B'$. Given any other channel $N$ mediating information between $A''$ and $B''$, the tensor product $M\otimes N$ associated with $A=A'A''$ and $B=B'B''$ is a valid channel. 

Such tensor products are crucial in information theory, as one often studies tasks in the asymptotic setting, where the same resource is used arbitrarily many times. Out of interests, for example in quantum gravity and in particular quantum black holes, we want to study information communication theory of processes with indefinite causal structure \cite{jia2017quantum}. We need to define products of processes and check if they are valid processes. Analogous to channel products, for two processes $W^{a'b'\cdots}$ and $Z^{a''b''\cdots}$ with the same number of parties, define their product $P^{ab\cdots}=W^{a'b'\cdots}\otimes Z^{a''b''\cdots}$. It takes in channels of parties $A,B,\cdots$ and outputs probabilities. Here $A=A'A''$, $B=B'B'', \cdots$. The situation for two parties is illustrated in Figure \ref{fig:pp}.

\begin{figure}
    \centering
\resizebox{.6\textwidth}{.225\textwidth}{
    \begin{tikzpicture}
\draw[dashed, thick, blue, fill=blue!30] (3,-8.5) node (v5) {} -- (2,-7.5) -- (0,-7.5) -- (0,-9.5) -- (1,-10.5) -- (v5);
\draw[dashed, thick, blue, fill=blue!30]  (1,-8.5) rectangle (3,-10.5);
\node at (2,-8) {\huge $\mathsf{b'_2}$};
\node at (2.5,-9) {\huge $B'$};

\draw[thick, fill=yellow!70] (3,-6) -- (2,-5) -- (-5,-5) -- (-5,-6.5) -- (-4,-7.5);
\draw[thick, fill=yellow!70] (-4,-13) -- (-5,-12) -- (-5,-10.5) -- (-2.5,-10.5) -- (-2.5,-7.5)--(-1.5,-7.5)--(-1.5,-11.5);
\draw[thick, fill=yellow!70] (3,-11.5) -- (2,-10.5) -- (-0.5,-10.5) -- (0.5,-11.5);
\draw[thick, fill=yellow!70] (-4,-7.5) node (v1) {} -- (-1.5,-7.5) -- (-1.5,-11.5) -- (-4,-11.5) -- (-4,-13) node (v4) {} -- (3,-13) -- (3,-11.5) -- (0.5,-11.5) -- (0.5,-7.5) -- (3,-7.5) -- (3,-6) node (v2) {} -- (-4,-6) -- (-4,-7.5) node (v3) {};
\draw[thick] (-5,-5) -- (-4,-6);
\draw[thick] (-2.5,-10.5) -- (-1.5,-11.5);
\draw[thick] (-5,-10.5) -- (-4,-11.5);
\node at (-1,-5.5) {\huge $Z$};

\draw[-triangle 90, dashed] [thick]  (1.5,-8) to (1.5,-7);
\node at (2,-11) {\huge $\mathsf{b'_1}$};
\draw[-triangle 90, dashed] [thick]  (1.5,-11) to (1.5,-10);

\draw[dashed, thick, red, fill=red!30]  (-4,-10.5)-- (-5,-9.5)--(-5,-7.5)--(-3,-7.5)--(-2,-8.5);
\draw[dashed, thick, red, fill=red!30]  (-4,-8.5) rectangle (-2,-10.5);
\draw[dashed, thick, red] (-5,-7.5) -- (-4,-8.5);
\draw[-triangle 90, dashed] [thick]  (-3.5,-8) to (-3.5,-7);
\draw[-triangle 90, dashed] [thick]  (-3.5,-11) to (-3.5,-10);
\node at (-4,-11) {\huge $\mathsf{a''_1}$};
\node at (-4,-8) {\huge $\mathsf{a''_2}$};
\node at (-4.5,-9) {\huge $A''$};

\draw[dashed, thick, blue, fill=blue!30] (4.5,-10) node (v5) {} -- (3.5,-9) -- (1.5,-9) -- (1.5,-11) -- (2.5,-12) -- (v5);
\draw[dashed, thick, blue, fill=blue!30]  (2.5,-10) rectangle (4.5,-12);
\node at (3.5,-9.5) {\huge $\mathsf{b'_2}$};
\node at (4,-10.5) {\huge $B'$};

\draw[thick, fill=yellow!70] (4.5,-7.5) -- (3.5,-6.5) -- (-3.5,-6.5) -- (-3.5,-8) -- (-2.5,-9);
\draw[thick, fill=yellow!70] (-2.5,-14.5) -- (-3.5,-13.5) -- (-3.5,-12) -- (-1,-12) -- (-1,-9)--(0,-9)--(0,-13);
\draw[thick, fill=yellow!70] (4.5,-13) -- (3.5,-12) -- (1,-12) -- (2,-13);
\draw[thick, fill=yellow!70] (-2.5,-9) node (v1) {} -- (0,-9) -- (0,-13) -- (-2.5,-13) -- (-2.5,-14.5) node (v4) {} -- (4.5,-14.5) -- (4.5,-13) -- (2,-13) -- (2,-9) -- (4.5,-9) -- (4.5,-7.5) node (v2) {} -- (-2.5,-7.5) -- (-2.5,-9) node (v3) {};
\draw[thick] (-3.5,-6.5) -- (-2.5,-7.5);
\draw[thick] (-1,-12) -- (0,-13);
\draw[thick] (-3.5,-12) -- (-2.5,-13);
\node at (0.5,-7) {\huge $W$};

\draw[-triangle 90, dashed] [thick]  (3,-9.5) to (3,-8.5);
\node at (3.5,-12.5) {\huge $\mathsf{b'_1}$};
\draw[-triangle 90, dashed] [thick]  (3,-12.5) to (3,-11.5);

\draw[dashed, thick, red, fill=red!30]  (-2.5,-12)-- (-3.5,-11)--(-3.5,-9)--(-1.5,-9)--(-0.5,-10);
\draw[dashed, thick, red, fill=red!30]  (-2.5,-10) rectangle (-0.5,-12);
\draw[dashed, thick, red] (-3.5,-9) -- (-2.5,-10);
\draw[-triangle 90, dashed] [thick]  (-2,-9.5) to (-2,-8.5);
\draw[-triangle 90, dashed] [thick]  (-2,-12.5) to (-2,-11.5);
\node at (-2.5,-12.5) {\huge $\mathsf{a'_1}$};
\node at (-2.5,-9.5) {\huge $\mathsf{a'_2}$};
\node at (-3,-10.5) {\huge $A'$};

\draw[thick, fill=yellow!70] (9.5,-14.5) -- (7,-12) -- (7,-10.5) -- (9.5,-10.5) -- (10,-8)--(12,-9)--(12,-13);
\draw[thick, fill=yellow!70] (16.5,-13) -- (14,-10.5) -- (12.5,-10.5) -- (14,-13);

\draw[dashed, thick, blue, fill=blue!30] (16.5,-10) node (v5) {} -- (15.5,-9) -- (13.5,-9) -- (13.5,-11) -- (14.5,-12) -- (v5);
\draw[dashed, thick, blue, fill=blue!30]  (14.5,-10) rectangle (16.5,-12);
\node at (15,-9.5) {\huge $\mathsf{b_2}$};
\node at (15.5,-10) {\huge $B$};

\draw[thick] (9.5,-10.5) -- (12,-13);
\draw[thick] (7,-10.5) -- (9.5,-13);
\node at (13,-11) {\huge $W$};

\node at (15,-11.5) {\huge $\mathsf{b_1}$};
\draw[-triangle 90, dashed] [thick]  (14.25,-12) to (14.25,-11);

\draw[dashed, thick, red, fill=red!30]  (9.5,-12)-- (7,-9.5)--(7,-7.5)--(9,-7.5)--(11.5,-10);
\draw[dashed, thick, red, fill=red!30]  (9.5,-10) rectangle (11.5,-12);
\draw[dashed, thick, red] (7,-7.5) -- (9.5,-10);
\draw[-triangle 90, dashed] [thick]  (9,-12) to (9,-11);
\node at (8.5,-11.5) {\huge $\mathsf{a_1}$};
\node at (8.5,-8.5) {\huge $\mathsf{a_2}$};
\node at (8.5,-10) {\huge $A$};

\draw[thick, fill=yellow!70] (16.5,-7.5) -- (14,-5) -- (7,-5) -- (7,-6.5) -- (9.5,-9);
\draw[thick, fill=yellow!70] (9.5,-9) node (v1) {} -- (12,-9) -- (12,-13) -- (9.5,-13) -- (9.5,-14.5) node (v4) {} -- (16.5,-14.5) -- (16.5,-13) -- (14,-13) -- (14,-9) -- (16.5,-9) -- (16.5,-7.5) node (v2) {} -- (9.5,-7.5) -- (9.5,-9) node (v3) {};
\draw[thick] (7,-5) -- (9.5,-7.5);

\draw[-triangle 90, dashed] [thick]  (9,-9) to (9,-8);

\draw[-triangle 90, dashed] [thick]  (14.25,-9.25) to (14.25,-8.25);
\node at (11.5,-6) {\huge $P$};

\draw[dashed, thick, blue] (14.5,-10) -- (14,-9.5);
\draw[dashed, thick, blue] (2.5,-10) -- (2,-9.5);

\draw[-triangle 90] [ultra thick] (5,-9.5) -- (6,-9.5);
\end{tikzpicture}
}
    \caption{Process product}
    \label{fig:pp}
\end{figure}
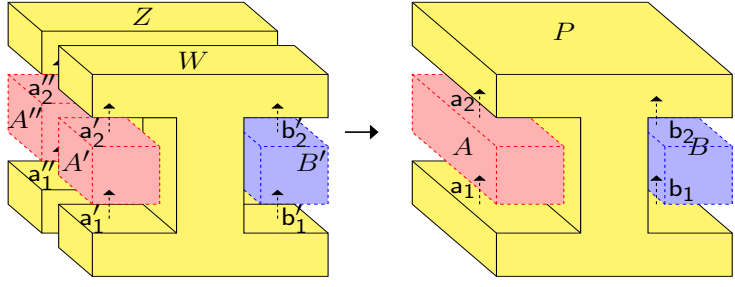

%The adaptation of the asymptotic setting in the processes framework needs some qualifications. First, in the context of processes a party represents a localized region of spacetime. If multiple copies of a processes is used, this introduces multiple copies of the parties. [[ ]] The most natural way to introduce asymptotic setting is to assume that the same parties share multiple copies of the process. [[]] Second, as we show below, not all processes allow such copying. 

A simple example shows that process products are not always valid processes. Consider the process
\begin{align}\label{eq:rqs}
W^{xy}=\frac{d_O}{2}(\omega^{x_1}\otimes\rho_{\sim}^{x_2y_1}\otimes\omega^{y_2}+\omega^{x_2}\otimes\rho_{\sim}^{x_1y_2}\otimes\omega^{y_1}),
\end{align}

%\begin{figure}[ht!]
%\begin{center}
%\includegraphics{example}\\
%\caption{\label{fig:W^{ab}} W^{ab} }
%\end{center}
%\end{figure}

where $d_O$ is the dimension of the outputs, $\omega$ is the maximally mixed state and $\rho_{\sim}:=(\id+\sigma_3\otimes \sigma_3)/4$ is a maximally correlating state that can represent a classical identity channel in the $\{\ket{0},\ket{1}\}$ basis. This process can be viewed as an equal-weight classical mixture of a channel from $x$ to $y$ and another one from $y$ to $x$. Suppose $A'$ and $B'$ share a process $W^{a'b'}$ of this form, and $A''$ and $B''$ also share a process $W^{a''b''}$ of the same form. The operator $W^{ab}:=W^{a'b'}\otimes W^{a''b''}$ for the two parties $A=A'A''$ and $B=B'B''$ is not a valid process. $\rho_{\sim}^{a_2'b_1'}\otimes \rho_{\sim}^{a_1''b_2''}$ includes a term $\sigma_3^{a_2'}\otimes\sigma_3^{a_1''}\otimes \sigma_3^{b_1'}\otimes \sigma_3^{b_2''}$, which leads to a type $a_1a_2b_1b_2$ term and according to Proposition \ref{prop:og} renders the process $W^{ab}$ invalid. Intuitively, $\rho_{\sim}^{a_2'b_1'}\otimes \rho_{\sim}^{a_1''b_2''}$ creates a causal loop and violates the normalization of probability condition.

%\begin{figure}[ht!]
%\begin{center}
%\includegraphics{issue}\\
%\caption{\label{fig:loop} Porblematic term }
%\end{center}
%\end{figure}

Although $W^{a'b'}$ and $W^{a''b''}$ cannot be composed directly, it is possible to have a global process $P^{ab}$ that reduces to the two individual processes upon partial tracing. For example, let $A$ and $B$ have a process of the same form of (\ref{eq:rqs}). This is a process on the combined parties. The reduced processes $\Tr_{a''b''}W^{ab}$ and $\Tr_{a'b'}W^{ab}$ are exactly $W^{a'b'}$ and $W^{a''b''}$. 

The restriction of process products has an analogy with the ``non-separability'' of entangled states. If $\rho^{xy}$ is entangled, then $\rho^{xy}\ne \rho^x\otimes \rho^y$. Similarly for some processes $W^{ab}\ne W^{a'b'}\otimes W^{a''b''}$. %The difference is that we break the global system into two parts in different ways. For $xy$ there is originally one party and we break it into two parties $x$ and $y$. For $ab$ there are originally two parties $a$ and $b$ and we break it into two pairs of parties. The first pair is $a_1b_1$ and the second is $a_2b_2$. 
The difference is that for processes tensor products not only may not recover the original process, but may even be invalid.

Note that the processes in the example is can be viewed as classical because one can regard it as a classical mixture of classical resources. One can also substitute Choi states of quantum channels for those of classical channels to obtain an example of quantum process that is restricted in forming products. The invalidity of arbitrary products is a feature of quantum as well as classical theory of processes with indefinite causal structure.

In general, we have the following necessary and sufficient condition for when two multipartite processes cannot (and can) be composed into a valid product process.
\begin{proposition}\label{prop:ppc}
A product $P=W\otimes Z$ of two processes $W$ and $Z$ is not a valid process if and only if there exist a nontrivial term of $W$ and a nontrivial term of $Z$ that obey: 1) On any party where one term is trivial, the other is either trivial or includes the out type. 2) On any party where one term is the in type, the other term includes the out type.
\end{proposition}
\begin{proof}
Suppose $W$ and $Z$ satisfy the conditions and consider the tensor product of the two nontrivial terms. By Proposition \ref{prop:og}, to prove that $P$ is invalid we need to show that $P$ contains a nontrivial term that is not type $a_1$ for any party $A$. Conditions 1) and 2) guarantee that this is satisfied for the product term we consider.

Conversely, suppose $P$ is not a valid process. By Proposition \ref{prop:og}, $P$ contains a nontrivial term that is not type $a_1$ for any party $A$. This term must arise out of a tensor product of nontrivial terms. On any $A$ this term of $P$ is type trivial, $a_1a_2$ or $a_2$. We consider each case in turn. For any $A$ where this term is trivial, it must come from a tensor product of terms that are trivial on $A$. Write this kind of tensor product as $(0,0)$, where $0$ denotes the trivial type. Next, for any $A$ where this term is type $a_1a_2$, it must come from a tensor product of the kind $(0,12)$, $(1, 12)$, $(1,2)$ or $(2,12)$, where $1$ and $2$ denote in and out types respectively. Finally, for any $A$ where this term is type $a_2$, it must come from a tensor product of the kind $(0,2)$ or $(2,2)$. To sum up, on any $A$, this term of $P$ comes from a tensor product of the kind $(0,0)$, $(0,12)$, $(0,2)$, $(1, 12)$, $(1,2)$, $(2,12)$,  or $(2,2)$. This implies conditions 1) and 2).
\end{proof}

A useful special case is the condition on two-party processes. Intuitively, the fulfillment of the two conditions in the corollary below give rise to causal loops, which violate the normalized probability condition for processes and hence lead to invalid products.
\begin{corollary}
A product $P^{ab}=W^{a'b'}\otimes Z^{a''b''}$ of two-party processes is not a valid process if and only if: 1) Both $W$ and $Z$ have signalling terms; 2) The Hilbert-Schmidt terms of $W$ and $Z$ put together contain signalling terms of both directions.
\end{corollary}
\begin{proof}
Suppose $W$ and $Z$ obey the two conditions. Then we can pick a signalling term from $W$ of one direction and a signalling term from $Z$ of the other direction. We show that this pair of terms satisfy conditions 1) and 2) in Proposition \ref{prop:ppc}, and hence the product is not a valid process. Neither term is trivial on either of the two parties, so 1) of Proposition \ref{prop:ppc} is fulfilled. 2) is also fulfilled because the terms signal to different directions.

Conversely, suppose $P$ is not valid. By Proposition \ref{prop:ppc}, there is a nontrivial term from $W$ and a nontrivial term from $Z$ that obey 1) and 2) of Proposition \ref{prop:ppc}. By Proposition \ref{prop:og}, both terms are in type on some party. By 2) of Proposition \ref{prop:ppc} they must be in type on different parties, and they include the out type on the parties where they are not in type. In other words, they are signalling terms to different directions. This proves conditions 1) and 2) of the statement.
\end{proof}

A product of more than two processes can be constructed iteratively, and the validity of the product process must be checked at each step. If a set of processes cannot form a valid product in one sequence of construction, changing the sequence of construction will not make it valid. This is because the invalid term will be always be present.

\section{Discussions}
We showed that for processes we cannot take tensor products unrestrictedly. For communication theory, our result implies that communication tasks defined in the asymptotic limit is not meaningful for processes characterized by Proposition \ref{prop:ppc} when local operation is unrestricted for the combined parties $A=A'A''\cdots$, $B=B'B''\cdots$, $\cdots$.\footnote{If local operations inside the combined party are assumed to be uncorrelated among the individual parties, the product process can be viewed as valid. In this case one might as well treat the individual processes separately without forming products. In communication theory this can correspond to repeated but uncorrelated uses of resources.} For these processes at least it is more suitable to study one-shot capacities. Similarly caution needs to be taken for asymptotic entanglement theory of processes \cite{jia2017quantum}.

The restriction induces some interesting questions. We studied the restriction within the process framework. To what extent does the restriction hold for indefinite causal structure theories in general \cite{hardy_probability_2005, hardy2007towards}? For the particular example we used to demonstrate the restriction, there exists a global process that reduces to the two individual processes. When is this true in general?

\section*{Acknowledgements}
J. D. thanks Lucien Hardy and Achim Kempf for useful discussions. N. S. thanks Lucien Hardy for useful discussions. Research at Perimeter Institute is supported by the Government of Canada through Industry Canada and by the Province of Ontario through the Ministry of Economic Development and Innovation. 

\bibliographystyle{apsrev}
\bibliography{bib.bib}

\end{document}